\documentclass{amsart}
\usepackage{amsfonts}
\usepackage{amsmath}
\usepackage{amssymb}

\theoremstyle{plain}
\newtheorem{theorem}{Theorem}[section]
\newtheorem{lemma}[theorem]{Lemma}
\newtheorem{proposition}[theorem]{Proposition}
\theoremstyle{definition}
\newtheorem{definition}[theorem]{Definition}

\theoremstyle{remark}

\newtheorem{remark}[theorem]{Remark}
\numberwithin{equation}{section}

\input{tcilatex}

\begin{document}
\title{The Lax integrability of a two-component hierarchy of the Burgers
type dynamical systems within asymptotic and differential-algebraic
approaches}
\author{Denis L. Blackmore}
\address{The Department of Mathematical Sciences at the NJIT, Newark, USA}
\email{denblac@gmail.com}
\author{Anatolij K. Prykarpatski}
\address{The Department of Applied Mathematics at AGH University of Science
and Technology of Krakow, Poland, and\ the Ivan Franko State Pedagogical
University of Drohobych, Lviv region, Ukraine}
\author{Emin \"{O}z\c{c}a\u{g}}
\address{The Department of Mathematics at the Hacettepe University of
Ankara, Turkey}
\email{sultan\_kamal@hotmail.com, ozcag1@hacettepe.edu.tr}
\author{Kamal Soltanov}
\address{The Department of Mathematics at the Hacettepe University of
Ankara, Turkey}
\email{sultan\_kamal@hotmail.com, ozcag1@hacettepe.edu.tr}
\email{pryk.anat@ua.fm, prykanat@cybergal.com}
\dedicatory{Authors dedicate with honor  this work to \ their colleague and
friend, a brilliant ukrainian mathematician Professor Anatoliy M. Samoilenko
in occasion of his 75-birthday Jubille.}
\maketitle

\begin{abstract}
The Lax type integrability of a two-component polynomial Burgers type
dynamical system within a differential-algebraic approach is studied, its \
linear adjoint matrix Lax representation is constructed. A related recursion
operator and infinite hierarchy of Lax integrable nonlinear dynamical
systems of the Burgers-Korteweg-de Vries type are derived by means of the
gradient-holonomic technique, \ the corresponding Lax type representations
are presented.
\end{abstract}

\section{\protect\bigskip \protect\bigskip Introduction}

Recently a great deal of research articles \cite{TW,IS,MSY,MSY1,SW} were
devoted to classification of polynomial integrable dynamical systems on
smooth functional manifolds. In particular, in the articles \cite{TW} there
was presented a wide enough list of two-component polynomial dynamical
systems of Burgers and Korteweg-de Vries type, which either reduce by means
of some, in general nonlocal, change of variables to the respectively \
separable triangle Lax type integrable forms or transform to the completely
linearizable flows. Amongst these systems the authors of \cite{TW} singled
out the following two-component Burgers type dynamical system 
\begin{equation}
\left. 
\begin{array}{c}
u_{t}=\ u_{xx}+2uu_{x}+v_{x} \\ 
v_{t}=u_{x}v+uv_{x}%
\end{array}%
\right\} :=K[u,v]  \label{B1.1}
\end{equation}%
on a smooth of the Schwartz type functional manifold $M\subset C^{\infty }(%
\mathbb{R};\mathbb{R}^{2}),$ where $(u,v)^{\intercal }\in M,$ the subscripts$%
"x"$ and $"t"\ $ denote, respectively, the partial derivatives with respect
to the variables $x\in \mathbb{R}$ and $t\in \mathbb{R}_{+},\ $the latter
being the evolution parameter. Within this work we will \textit{a priori}
assume that the dynamical system \ (\ref{B1.1}) possesses smooth enough
solutions for an evolution parameter $t\in \mathbb{R}_{+},$  as it follows
from the standard functional-analytic compactness principle considerations
from \cite{Li}.

It is mentioned in \cite{TW} (p. 7706) that the Burgers type dynamical
system \ (\ref{B1.1})\ was before extensively studied in \cite{Fo,Ma}, where
"... the symmetry integrability of \ (\ref{B1.1}) as well the existence of a
recursion operator has already been demonstrated..." for it. The dynamical
system \ (\ref{B1.1}) appears to have interesting applications, as its
long-wave limit reduces to the well-known Leroux system \cite{FrTr,Wh},
describing dynamical processes in two-component hydro- and lattice gas
dynamics. As there is claimed in \cite{TW} (p. 7726), by now the
integrability of \ (\ref{B1.1}) remains still unproven, and having found no
new result devoted to this problem available in literature, we undertaken
this challenge to close it by means of the gradient-holonomic \cite{PM,BPS},
linear adjoint mapping \cite{PrY} approaches \ and recently devised \cite%
{PAPP,PAPP1} differential algebraic integrability testing tools. As a
general result we have proved the following theorem.

\begin{theorem}
\label{Tm_B1.1} \label{Tm_B3.2 copy(1)} The two-component polynomial Burgers
type dynamical system \ (\ref{B1.1}) possesses only two local conserved
quantities $\int dxu$ and $\int dxv$ and no other infinite affine ordered
local conserved quantities. Moreover, on the functional manifold $M$ the
Burgers type dynamical system \ (\ref{B1.1}) is linearizable by means of a
Hopf-Cole type transformation and a dual adjoint mapping to the matrix Lax
type representation 
\begin{equation}
D_{x}\left( 
\begin{array}{c}
f \\ 
\hat{f}%
\end{array}%
\right) =\left( 
\begin{array}{cc}
(u+\lambda ^{-1}v-\lambda )/2 & 0 \\ 
1 & (\lambda -\lambda ^{-1}v-u\ )/2]%
\end{array}%
\right) \left( 
\begin{array}{c}
f \\ 
\hat{f}%
\end{array}%
\right) ,\text{\ \ \ }  \label{B1.2a}
\end{equation}%
and%
\begin{equation}
D_{t}\left( 
\begin{array}{c}
f \\ 
\hat{f}%
\end{array}%
\right) =\left( 
\begin{array}{cc}
\begin{array}{c}
D_{x}(u+\lambda )]/2+ \\ 
+(u+\lambda )(u+\lambda ^{-1}v-\lambda )/2%
\end{array}
& 0 \\ 
(u+\lambda ) & 
\begin{array}{c}
-D_{x}(u+\lambda )]/2+ \\ 
+(u+\lambda )(\lambda -u-\lambda ^{-1}v)/2%
\end{array}%
\end{array}%
\right) \left( 
\begin{array}{c}
f \\ 
\hat{f}%
\end{array}%
\right) ,  \label{B1.2b}
\end{equation}%
compatible for all $\lambda \in \mathbb{C}\backslash \{0\}$ with $(f,\ \hat{f%
})^{\intercal }\in \Lambda ^{0}(\mathcal{\bar{K}\{}u,v;D_{x}^{-1}\sigma
|N\})^{2},\ $\ where $\mathcal{\bar{K}\{}u,v;D_{x}^{-1}\sigma |N\}$ denotes
some  \cite{POS}  finitely extended differential ring $\mathcal{\bar{K}\{}%
u,v\}.$ The related with the Lax operator \ (\ref{B1.2a}) infinite hierarchy
of generalized Burgers type dynamical systems allows the following compact
representation:%
\begin{equation}
D_{t_{n}}(u+\lambda ^{-1}v-\lambda )=D_{x}[D_{x}^{\ \ }\alpha _{n}(x;\lambda
)\ \ +\ (u+\lambda ^{-1}v-\lambda )\alpha _{n}(x;\lambda )],  \label{B1.2c}
\end{equation}%
where the evolution parameters $t_{n}\in \mathbb{R}_{+}$ and, by definition, 
$\ $%
\begin{equation}
\alpha _{n}(x;\lambda ):=(\lambda ^{n}\alpha _{n}(x;\lambda ))_{+}
\label{B1.2d}
\end{equation}%
for all natural $n\in \mathbb{N}$ is the corresponding nonnegative degree
polynomial part generated by the asymptotic local functional solution $\
\alpha _{n}(x;\lambda )\sim \sum_{j\in \mathbb{Z}_{+}}\lambda ^{-j}\alpha
_{j}[u,v]$ as $|\lambda |\rightarrow \infty $ to the differential functional
equation 
\begin{equation}
D_{x}^{2}\alpha _{n}(x;\lambda )+D_{x}((u+\lambda ^{-1}v-\lambda )\alpha
_{n}(x;\lambda ))=0.  \label{B1.2e}
\end{equation}
\end{theorem}

As a simple consequence of the Theorem \ref{Tm_B1.1} one finds that the
Burgers type dynamical system \ (\ref{B1.1}) does not allow on the
functional manifold $M$ a Hamiltonian formulation, and the corresponding
recursion operator 
\begin{equation}
\text{\ }\Phi :=\left( 
\begin{array}{cc}
\ D_{x}+D_{x}^{\ }uD_{x}^{-1} & \text{ \ \ }1 \\ 
D_{x}^{\ }vD_{x}^{-1} & \text{ \ \ \ }0%
\end{array}%
\right) ,  \label{B1.2f}
\end{equation}%
satisfying the determining commutator equation%
\begin{equation}
D_{t}\Phi =[K^{\prime },\Phi ],  \label{B1.2g}
\end{equation}%
for the two-component Burgers type dynamical system \ (\ref{B1.1}) and found
before in \cite{Fo,Ma}, proves to be not factorizable by means of compatible
Poissonian structures, as they on the whole do not exist.

The scalar Lax type representation \ (\ref{B1.2a}) can be reduced by means
of the nonlocal change of variables \ $\tilde{g}=$ $\ f^{2}\exp [\lambda
^{2}t-D_{x}^{-1}(\lambda ^{-1}v-\lambda )]$  $\in \Lambda ^{0}(\mathcal{\bar{%
K}\{}u,v;D_{x}^{-1}\sigma |N\})$ \ to a more simpler linear form%
\begin{equation}
D_{t}\ \tilde{g}=D_{xx}\ \tilde{g}+v\ \tilde{g}=0,\ \ D_{x}\ \tilde{g}=\ u\ 
\tilde{g},  \label{B1.3}
\end{equation}%
compatible for $\tilde{g}\in \Lambda ^{0}(\mathcal{\bar{K}\{}%
u,v;D_{x}^{-1}\sigma |N\})$ and not depending on the parameter $\lambda \in 
\mathbb{C}\backslash \{0\}.$ The above representation \ (\ref{B1.3}) can be
easily generalized to the following higher-order evolution equation case:%
\begin{equation}
D_{t}\tilde{g}=D_{x}^{n}\tilde{g}+v\text{ }\tilde{g}=0,\text{ \ \ }D_{x}%
\tilde{g}=\ u\ \tilde{g}\ ,  \label{B1.4}
\end{equation}%
where $n\in \mathbb{Z}_{+}.$ Making use of the mentioned above nonlocal
change of variables \ $\hat{g}=$ $\tilde{g}\exp (\lambda ^{-1}D_{x}^{-1}v)$ $%
\in \Lambda ^{0}(\mathcal{\bar{K}\{}u,v;D_{x}^{-1}\sigma |N\}),$ one can
obtain a new infinite hierarchy of two-component Lax type integrable
polynomial Burgers type dynamical systems, generalizing those discussed
before in \ \cite{Ta,POS}. \ For instance, at $n=3$ \ we find the following
dynamical Burgers type dynamical system of the third order:%
\begin{eqnarray}
\ u_{t} &=&\ u_{3x}+3\ (uu_{x})_{x}+\ 3u^{2}u_{x}+v_{x},  \label{B1.5} \\
v_{t} &=&(ur[u,v])_{x},  \notag
\end{eqnarray}%
where $r:J[u,v]\rightarrow C^{\infty }(\mathbb{R}^{2};\mathbb{R}^{2})$ is a
polynomial mapping on the jet-space $J(\mathbb{R}^{2};\mathbb{R}^{2})\ $of
elements $(x,t;u,v,D_{x}u,D_{x}v,D_{t}u,D_{t}v,D_{x}^{2}u,D_{x}^{2}v,...)\in
J(\mathbb{R}^{2};\mathbb{R}^{2}),$ \ suitably determined by the relationship
\ (\ref{B1.2d}) at $n=2.$ Its scalar Lax type representation easily obtains,
respectively, either from \ (\ref{B1.2a}),\ (\ref{B1.2b}) or from \ (\ref%
{B1.4}). The latter at $n=3$ easily gives rise to the scalar Lax type
representation 
\begin{equation}
\begin{array}{c}
D_{x}\hat{g}=(u+\lambda ^{-1}v)\hat{g}, \\ 
D_{t}\hat{g}=\ D_{xxx}\hat{g}+3\lambda ^{-1}vD_{xx}\hat{g}+3(v_{x}/\lambda
+v^{2}/\lambda ^{2})D_{x}\hat{g}+ \\ 
+(v_{xx}/\lambda +3vv_{x}/\lambda ^{2}+v^{3}/\lambda ^{3}-(u\eta \lbrack
u,v])/\lambda )\hat{g}=0,\text{ \ \ }%
\end{array}
\label{B1.6}
\end{equation}%
compatible for all $\lambda \in \mathbb{C}\backslash \{0\}$ and $\hat{g}\in
\Lambda ^{0}(\mathcal{\bar{K}\{}u,v;D_{x}^{-1}\sigma |N\})^{2},$ which can
be suitably extended by means of the related adjoint mapping to the matrix
representation.

\section{Differential-algebraic preliminaries}

As our consideration of the integrability problem, discussed above, will be
based on some differential-algebraic techniques, to be for further more
precise, we need to involve here some additional differential-algebraic
preliminaries \cite{GMS,GD,GD1,GD2,Go,BPS}.

Take the ring $\mathcal{K}:=\mathbb{R}\{\{x,t\}\},$ $(x,t)\in \mathbb{%
R\times }(0,T\mathbb{)},$ of convergent germs of real-valued smooth
functions from $C^{\infty }(\mathbb{R}^{2};\mathbb{R})$ and construct the
associated differential quotient ring $\mathcal{K}\{u,v\}:=Quot(\mathcal{K}%
[\Theta u,\Theta v])$ with respect to two functional variables $u,v\in 
\mathcal{K},$ where $\Theta $ denotes \cite{Ka,Ri,GD0,GD,GMS} the standard
monoid of all commuting differentiations $D_{x}$ and $D_{t},$ satisfying the
standard Leibniz condition, and defined by the natural conditions 
\begin{subequations}
\begin{equation}
D_{x}(x)=1=D_{t}(t),\ \ \ D_{t}(x)=0=D_{x}(t),  \label{B2.1}
\end{equation}%
The ideal $I\{u,v\}\subset \mathcal{K}\{u,v\}$ is called differential if the
condition $I\{u,v\}=\Theta I\{u,v\}$ holds. In the differential ring $%
\mathcal{K}\{u,v\},$ interpreted as an invariant differential ideal in $%
\mathcal{K},$ there are two naturally defined differentiations 
\end{subequations}
\begin{equation}
D_{t},\ D_{x}:\ \mathcal{K}\{u,v\}\rightarrow \mathcal{K}\{u,v\},
\label{B2.2}
\end{equation}%
satisfying the commuting relationship%
\begin{equation}
\lbrack D_{t},\ D_{x}]=0.  \label{B2.3}
\end{equation}%
Consider the ring $\mathcal{K\{}u,v\},$ $u,v\in \mathcal{K},$ and the
exterior differentiation $d:\mathcal{K\{}u,v\}\rightarrow \Lambda ^{1}(%
\mathcal{K\{}u,v\}),$ $:d:\Lambda ^{p}(\mathcal{K\{}u,v\})\rightarrow
\Lambda ^{p+1}(\mathcal{K\{}u,v\})\ $for $p\in \mathbb{Z}_{+},$ acting in
the freely generated Grassmann algebras $\Lambda (\mathcal{K\{}u,v\})=\oplus
_{p\in \mathbb{Z}_{+}}\Lambda ^{p}(\mathcal{K\{}u,v\})$ over the field $%
\mathbb{C},$ \ where by definition,%
\begin{equation}
\begin{array}{c}
\Lambda ^{1}(\mathcal{K\{}u,v\}):=\ \ \mathcal{K\{}u,v\}dx+\ \mathcal{K\{}%
u,v\}dt+ \\ 
+\sum_{j,k\in \mathbb{Z}_{+}}\mathcal{K\{}u,v\}du^{(j,k)}+\sum_{j,k\in 
\mathbb{Z}_{+}}\mathcal{K\{}u,v\}dv^{(j,k)}, \\ 
u^{(j,k)}:=D_{t}^{j}D_{x}^{k}u,\text{ \ \ \ \ \ \ }%
v^{(j,k)}:=D_{t}^{j}D_{x}^{k}v, \\ 
\Lambda ^{2}(\mathcal{K\{}u,v\}):=\mathcal{K\{}u,v\}d\Lambda ^{1}(\mathcal{%
K\{}u,v\}),...,\text{ } \\ 
\Lambda ^{p+1}(\mathcal{K\{}u,v\}):=\mathcal{K\{}u,v\}d\Lambda ^{p}(\mathcal{%
K\{}u,v\}),%
\end{array}%
\text{ \ \ }  \label{B2.4}
\end{equation}%
The triple $\mathcal{A}:\mathcal{=}(\mathcal{K\{}u,v\},\Lambda (\mathcal{K\{}%
u,v\});d)$ will be called \textit{the Grassmann differential algebra} with
generatrices $u,v\in \mathcal{K}.$ In the algebra $\mathcal{A},$ generated
by $u,v\in \mathcal{K},$ one naturally defines $\ $the action of
differentiations $D_{t},D_{x}$ and $\ \partial /\partial u^{(j,k)},\partial
/\partial v^{(j,k)}:\mathcal{A}\rightarrow \mathcal{A},j,k\in \mathbb{Z}_{+},
$ as follows: 
\begin{equation}
\begin{array}{c}
D_{t}u^{(j,k)}=u^{(j+1,k)},D_{x}u^{(j,k)}=u^{(j,k+1)}, \\ 
D_{t}\ v^{(j,k)}=\ v^{(j+1,k)},D_{x}v^{(j,k)}=v^{(j,k+1)}, \\ 
D_{t}du^{(j,k)}=du^{(j+1,k)},D_{x}du^{(j,k)}=du^{(j,k+1)}, \\ 
D_{t}dv^{(j,k)}=dv^{(j+1,k)},D_{x}dv^{(j,k)}=dv^{(j,k+1)}, \\ 
dP[u,v]=\sum_{j,k\in \mathbb{Z}_{+}}du^{(j,k)}\wedge \partial
P[u,v]/\partial u^{(j,k)}+\sum_{j,k\in \mathbb{Z}_{+}}dv^{(j,k)}\wedge
\partial P[u,v]/\partial v^{(j,k)}= \\ 
=\sum_{j,k\in \mathbb{Z}_{+}}(\pm )\partial P[u,v]/\partial u^{(j,k)}\wedge
du^{(j,k)}+ \\ 
+\sum_{j,k\in \mathbb{Z}_{+}}(\pm )\partial P[u,v]/\partial v^{(j,k)}\wedge
dv^{(j,k)}\ :=<P^{\prime }[u,v],\wedge (du,dv)^{\intercal }>_{\mathbb{R}%
^{2}},%
\end{array}
\label{B2.5}
\end{equation}%
where the sign $"\wedge "$ denotes the standard \cite{Go} exterior
multiplication in $\Lambda (\mathcal{K\{}u,v\}),$ and for any $P[u,v]\in
\Lambda (\mathcal{K\{}u,v\})$ \ the mapping 
\begin{equation}
P^{\prime }[u,v]\wedge :\Lambda ^{0}(\mathcal{K\{}u,v\})^{2}\ \rightarrow
\Lambda (\mathcal{K\{}u,v\}),  \label{B2.5a}
\end{equation}%
is linear. Moreover, the commutation relationships 
\begin{equation}
D_{x}d=d\text{ }D_{x},\text{ \ \ \ }D_{t}d=d\text{ }D_{t}  \label{R5aaa}
\end{equation}%
hold in the Grassmann differential algebra $\mathcal{A}.$ \ The following
remark \cite{GMS} is also important.

\begin{remark}
Any Lie derivative $L_{V}:\mathcal{K\{}u,v\}\mathcal{\rightarrow K\{}u,v\},$
\ satisfying the condition $L_{V}:\mathcal{K}\mathcal{\subset }\mathcal{K},$
can be uniquely extended to the differentiation $L_{V}:\mathcal{A}%
\rightarrow \mathcal{A},$ satisfying the commutation condition $L_{V}d=d$ $%
L_{V}.$
\end{remark}

The \textit{variational derivative, or the functional gradient }$\nabla
P[u,v]\in \Lambda (\mathcal{K\{}u,v\})^{2}$ with respect to the variables $%
u,v\in \mathcal{K},$ \textit{is defined }for any $P[u,v]\in \Lambda (%
\mathcal{K\{}u,v\})$ \textit{by means of the following expression:} 
\begin{equation}
\mathrm{grad}P[u,v]=P^{\prime ,\ast }[u,v](1),  \label{B2.7}
\end{equation}%
where a mapping $P^{\prime ,\ast }[u,v]:\Lambda ^{0}(\mathcal{K\{}%
u,v\})\rightarrow \Lambda ^{0}(\mathcal{K\{}u,v\})^{2}$ is the formal
adjoint mapping for that of \ (\ref{B2.5a}). The latter is strongly based on
the following important lemma, stated for a special case in \cite%
{GMS,GD0,GD,GD1,GD2,Ol}.

\begin{lemma}
\label{Lm_B2.1}Let the differentiations $D_{x}$ and $D_{t}:\Lambda (\mathcal{%
K\{}u,v\})\rightarrow \Lambda (\mathcal{K\{}u,v\})$ satisfy the conditions \
(\ref{B2.5}). Then the mapping%
\begin{equation}
\begin{array}{c}
Ker\mathrm{grad}/(\mathrm{Im}d\oplus \mathbb{C)\simeq }H^{1}(\mathcal{A}):=
\\ 
\\ 
=Ker\{d:\Lambda ^{1}(\mathcal{K\{}u,v\})\rightarrow \Lambda ^{2}(\mathcal{K\{%
}u,v\})\}/d\Lambda ^{0}(\mathcal{K\{}u,v\})\ 
\end{array}
\label{B2.8}
\end{equation}%
is a canonical isomorphism, where $H^{1}(\mathcal{A})\ $ is the
corresponding cohomology class of the Grassmann complex $\Lambda (\mathcal{%
K\{}u,v\}).$
\end{lemma}

It is well known \cite{Ri} that in the case of the differential ring $%
\mathcal{K\{}u,v\}$ not all of the cohomology classes $H^{j}(\mathcal{A}%
),j\in \mathbb{Z}_{+},$ are trivial. \ Nonetheless, one can impose on the
functions $u,v\in \mathcal{K}$ some additional restrictions, which will give
rise to the condition $H^{1}(\mathcal{A}),$ or equivalently, to the
relationship $Ker\nabla =$ $\mathrm{Im}D_{x}\oplus \mathrm{Im}D_{t}\oplus 
\mathbb{C}.$ In addition, the following simple relationship will hold:%
\begin{equation}
\mathrm{grad}\ (\mathrm{Im}D_{x}\oplus \mathrm{Im}D_{t})=0.  \label{B2.9}
\end{equation}

Based on Lemma \ \ref{Lm_B2.1} one can \ define the equivalence class $%
\widetilde{\mathcal{A}}:=\mathcal{A}/\{\mathrm{Im}D_{x}\oplus \mathrm{Im}%
D_{t}\oplus \mathbb{R\}}$ $:\mathbb{=}\mathcal{D(A};dxdt),$ whose elements
will be called \textit{functionals, }that is any element $\gamma \in $ $%
\mathcal{D(A};dxdt)$ can be\textit{\ } represented as a suitably defined
integral $\gamma :=\int \int dxdt\gamma \lbrack u,v]\in $ $\mathcal{D(A}%
;dxdt)$ for some $\gamma \lbrack u,v]\in \Lambda (\mathcal{K\{}u,v\})$ with
respect to the Lebesgue measure $dxdt$ on $\mathbb{R}^{2}.$

Consider now our two-component dynamical system \ (\ref{B1.1}) as a
polynomial differential constraint%
\begin{equation}
D_{t}(u,v)^{\intercal }=K[u,v],\   \label{B2.10}
\end{equation}%
imposed on the ring $\mathcal{K\{}u,v\}.$ The following definitions will be
useful for our further analysis.

\begin{definition}
Let the reduced ring $\mathcal{\bar{K}}\{u,v\}:=\mathcal{K}%
\{u,v\}|_{D_{t}(u,v)^{\intercal }=K[u,v]}\ .$ Then the triple \textit{\ }$%
\overline{\mathcal{A}}:=(\mathcal{\bar{K}}\{u,v\},\Lambda (\mathcal{\bar{K}}%
\{u,v\}),d)$ will be called a reduced Grassmann differential algebra over
the reduced ring $\mathcal{\bar{K}}\{u,v\}.$
\end{definition}

\begin{definition}
Any pair of elements $(\gamma \lbrack u,v],\rho \lbrack u,v])^{\intercal
}\in \Lambda ^{0}(\mathcal{\bar{K}\{}u,v\})^{2},$ satisfying the
relationship 
\begin{equation}
D_{t}\gamma \lbrack u,v]+D_{x}\rho \lbrack u,v]=0,  \label{B2.10a}
\end{equation}%
is called a scalar conservative quantity with respect to the
differentiations $D_{x}$ and $D_{t}.$
\end{definition}

Based on the differential-algebraic setting, described above, one can
naturally define the \ spaces of functionals $\mathcal{D(}\overline{\mathcal{%
A}};dx):=$ $\overline{\mathcal{A}}/\{D_{x}\overline{\mathcal{A}}\mathcal{\}}$
and $\mathcal{D(}\overline{\mathcal{A}};dt)=$ $\overline{\mathcal{A}}/\{D_{t}%
\overline{\mathcal{A}}\mathcal{\}}$ on the \textit{the reduced Grassmann
differential algebra }$\overline{\mathcal{A}}.$ From the functional point of
view these \ factor spaces $\mathcal{D(\overline{\mathcal{A}}};dx)$ and $%
\mathcal{D(\overline{\mathcal{A}}};dt)$ can be understood more classically
as the corresponding spaces of suitably defined integral expressions subject
to the measures $dx$ and $dt,$ respectively. Then the relationship (\ref%
{B2.10a}) means equivalently that the functional $\gamma :=\int dx\gamma
\lbrack u,v]\in \mathcal{D(}\overline{\mathcal{A}};dx)$ is a conserved
quantity for the differentiation $D_{t},$ and the functional $\Upsilon
:=\int dt\rho \lbrack u,v]\in \mathcal{D(}\overline{\mathcal{A}};dt)$ is a
conserved quantity for the differentiation $D_{x}.$

Since the differential relationship \ (\ref{B2.10}) naturally defines \cite%
{GMS,Go} on the reduced ring $\mathcal{\overline{K}\{}u,v\mathcal{\}}$ a
smooth vector field $K:\mathcal{\overline{K}\{}u,v\mathcal{\}\rightarrow }T%
\mathcal{(\overline{K}\{}u,v\mathcal{\})},$ one can construct the
corresponding Lie derivative $L_{K}:\overline{\mathcal{A}}\rightarrow 
\overline{\mathcal{A}}$ \ along this vector field and calculate the
differential Lax type \cite{La} expression

\begin{equation}
\partial \varphi \lbrack u,v]/\partial t+L_{K}\varphi \lbrack u,v]=0
\label{B2.10aa}
\end{equation}%
$\ $ for the element $\varphi \lbrack u,v]:=\mathrm{grad}\gamma \lbrack
u,v]\in \Lambda ^{0}(\mathcal{\bar{K}\{}u,v;D_{x}^{-1}\sigma |N\})^{2},$
where $\mathcal{\bar{K}\{}u,v;D_{x}^{-1}\sigma |N\}$ denotes some  finitely
extended differential ring $\mathcal{\bar{K}\{}u,v\}\ $and  $\gamma \in 
\mathcal{D(}\overline{\mathcal{A}};dx)$ is an arbitrary scalar conserved
quantity with respect to the differentiation $D_{t}.$   The following
classical Noether-Lax lemma \cite{La,Ol,BPS,PM,Ol}, inverse to the Lax
relationship \ (\ref{B2.10aa}), holds.

\begin{lemma}
\label{Lm_B2.2}\textbf{\ (\textit{E.Noether-P.Lax}) }Let a quantity $\varphi
\lbrack u,v]\in \Lambda ^{0}(\mathcal{\bar{K}\{}u,v;D_{x}^{-1}\sigma
|N\})^{2}$ be such that the following equation%
\begin{equation}
D_{t}\varphi \lbrack u,v]+K^{\prime ,\ast }[u,v]\varphi \lbrack u,v]=0,
\label{B2.10b}
\end{equation}%
equivalent to \ (\ref{B2.10aa}), holds \ in the ring $\mathcal{\bar{K}\{}%
u,v;D_{x}^{-1}\sigma |N\}\ $ satisfying the differential constraint \ (\ref%
{B2.10}). Then, if the Volterra condition $\varphi ^{\prime ,\ast
}[u,v]=\varphi ^{\prime }[u,v]\ \ $is satisfied in the ring$\mathcal{\ \bar{K%
}\{}u,v;D_{x}^{-1}\sigma |N\},$ the constructed homology type functional 
\begin{equation}
\gamma :=\int_{0}^{1}d\lambda \int dx<\varphi \lbrack \lambda u,\lambda
v],(u,v)^{\intercal }>_{\mathbb{C}^{2}}\in \mathcal{D(}\overline{\mathcal{A}}%
;dx)  \label{B2.10c}
\end{equation}%
is a scalar conserved quantity with respect to the differentiation $D_{t}.$
\end{lemma}

Assume now that the nonlinear two-component polynomial dynamical system \ (%
\ref{B2.10}) possesses a nontrivial compatible differential Lax type
representation in the form 
\begin{equation}
D_{x}f(x,t;\lambda )=l[u,v;\lambda ]f(x,t;\lambda ),\text{ \ \ \ \ \ \ }%
D_{t}f(x,t;\lambda )=p[u,v;\lambda ]f(x,t;\lambda )  \label{B2.11}
\end{equation}%
for some $\ $matrices$\ \ l[u,v;\lambda ],p[u,v;\lambda ]\in End$ $\Lambda
^{0}(\mathcal{\bar{K}\{}u,v\})^{q},$ $f(x,t;\lambda )\in \Lambda ^{0}(%
\mathcal{\bar{K}\{}u,v;D_{x}^{-1}\sigma |N\})^{q},$ analytically depending
on a parameter $\lambda \in \mathbb{C},$ where $q\in \mathbb{Z}%
_{+}\backslash \{0,1\}$ is finite. Then the following important proposition,
based on the gradient-holonomic approach, devised before in \cite{BPS,PM},
holds.

\begin{proposition}
\label{Prop_B2.3}\bigskip The Lax type integrable dynamical system \ (\ref%
{B2.10}) possesses a set (either finite or infinite) of naturally ordered
functionally independent scalar conserved differential quantities 
\begin{equation}
D_{t}\sigma _{j}[u,v]+D_{x}\rho _{j}[u,v]=0,  \label{B2.12}
\end{equation}%
where the pairs $(\sigma _{j}[u,v],\rho _{j}[u,v])^{\intercal }\in \Lambda
^{0}(\mathcal{\bar{K}\{}u,v\})^{2},j\in \mathbb{Z}_{+}.$
\end{proposition}

\begin{proof}
Assume that the Lax type integrable dynamical system \ (\ref{B2.10})
possesses a set (either finite or infinite) of naturally ordered
functionally independent scalar conserved differential quantities \ (\ref%
{B2.12}). Let $\mathcal{\bar{K}\{}u,v;D_{x}^{-1}\sigma |N\}$ denote the
finitely extended differential ring $\mathcal{K\{}u,v;\{D_{x}^{-1}\sigma
_{j}[u,v]:j=\overline{0,N}\}\}$ for arbitrary finite integer $N\in \mathbb{Z}%
_{+}$ \ under\ the constraints (\ref{B2.10}). Then the Lax equation \ (\ref%
{B2.10b}), if considered on the invariant functional submanifold 
\begin{eqnarray}
M_{N} &:&=\{(u,v)^{\intercal }\in M:\text{ \ \ }\mathrm{grad}<c^{(N)},\int dx%
\text{ }\Sigma ^{(N)}>_{\mathbb{C}^{N+1}}=0,  \label{B2.12a} \\
c^{(N)} &\in &\mathbb{C}^{N+1}\backslash \{0\},\Sigma ^{(N)}:=(\sigma
_{0},\sigma _{1},...,\sigma _{N})^{\intercal }\in \Lambda ^{0}(\mathcal{\bar{%
K}\{}u,v\})^{N+1}\},  \notag
\end{eqnarray}%
allows \cite{BPS,PM} as $|\lambda |\rightarrow \infty $ an asymptotic
solution $\varphi (x;\lambda )\in \Lambda ^{0}(\mathcal{\bar{K}\{}%
u,v;D_{x}^{-1}\sigma |N\})^{2}$ \ in the form $\ $ 
\begin{equation}
\varphi (x;\lambda )\ \sim \psi (x,t;\lambda )\ \exp \{\omega (x,t;\lambda
)+D_{x}^{-1}\sigma (x,t;\lambda )\},  \label{B2.13}
\end{equation}%
with a scalar analytical "dispersion" function $\omega (x,t;\cdot ):\mathbb{C%
}\rightarrow \mathbb{C}\ \ $ \ determined for all $(x,t)\in \mathbb{R}\times
\lbrack 0,T),$ and the compatible local functionals expansions $\ $%
\begin{eqnarray}
\Lambda ^{0}(\mathcal{\bar{K}\{}u,v\})^{2} &\ni &\ \sigma (x,t;\lambda )\sim
\sum_{j\in \mathbb{Z}_{+}}\sigma _{j}[u,v]\lambda ^{-j+|\sigma |},\text{ }
\label{B2.14} \\
\text{\ }\Lambda ^{0}(\mathcal{\bar{K}\{}u,v\})^{2} &\ni &\psi (x,t;\lambda
)\sim \sum_{j\in \mathbb{Z}_{+}}\psi _{j}[u,v]\lambda ^{-j+|\psi |}  \notag
\end{eqnarray}%
for some fixed integers $|\sigma |,|\psi |\in \mathbb{Z}_{+}.$ Moreover,
owing to the Lax equation (\ref{B2.10b}), all of the scalar functionals 
\begin{equation}
\gamma _{j}:=\int dx\sigma _{j}[u,v]  \label{B2.15}
\end{equation}%
for $j\in \mathbb{Z}_{+}$ are conserved quantities with respect to the
differentiation $D_{t}.$ Now, vice versa, if the Lax equation (\ref{B2.10b})
possesses an asymptotic as $|\lambda |\rightarrow \infty $ solution in the
form \ (\ref{B2.13}) $\varphi \lbrack u,v;\lambda ]\in \Lambda ^{0}(\mathcal{%
\bar{K}\{}u,v;D_{x}^{-1}\sigma |N\})^{2}$ \ \ with compatible expansions \ (%
\ref{B2.14}), then all of the scalar functionals \ (\ref{B2.15}) are, a
priori, the conserved quantities with respect to the differentiation $D_{t},$
that is there exist such scalar quantities $\rho _{j}[u,v]\in \Lambda ^{0}(%
\mathcal{\bar{K}\{}u,v\}),j\in \mathbb{Z}_{+},$ satisfying the relationships
\ (\ref{B2.12}).
\end{proof}

The analytical expressions for representation \ (\ref{B2.13}) and asymptotic
expansions \ (\ref{B2.14}) for a Lax type integrable dynamical system \ (\ref%
{B2.10}) easily enough follow from the general theory of asymptotic
solutions \cite{CL,Sh} to linear differential equations, applied to a linear
differential system \ (\ref{B2.11}) and from an important fact \cite%
{No,FT,BPS,PM}, that the trace functional $\Delta \lbrack u,v;\lambda
]:=tr(F(x,t;\lambda )C(\lambda )$ $\bar{F}(x,t;\lambda ))\in \Lambda ^{0}(%
\mathcal{\bar{K}\{}u,v;D_{x}^{-1}\sigma |N\})\ $with$\ $any$\ $\ constant
matrix $\ C(\lambda )\in \mathrm{End}$ $\mathbb{C}^{q}$ is for almost all $%
\lambda \in \mathbb{C}$ a conserved quantity with respect to both
differentiations $D_{t}$ and $D_{x},$ where $F(x,t;\lambda )$ and $\bar{F}%
(x,t;\lambda ),$ $(x,t)\in \mathbb{R}\times \mathbb{R}_{+},$ are,
respectively, the fundamental solutions to the linear Lax type equation 
\begin{equation}
D_{x}f(x,t;\lambda )=l[u,v;\lambda ]f(x,t;\lambda )\   \label{B2.16}
\end{equation}%
and its adjoint version 
\begin{equation}
D_{x}\bar{f}(x,t;\lambda )=-\bar{f}(x,t;\lambda )l[u,v;\lambda ],
\label{B2.16a}
\end{equation}%
where $f(x,t;\lambda ),\bar{f}^{\intercal }(x,t;\lambda )\in \Lambda ^{0}(%
\mathcal{\bar{K}\{}u,v;D_{x}^{-1}\sigma |N\})^{q}.$ Thereby, the
corresponding gradient 
\begin{equation}
\mathrm{grad}\Delta \lbrack u,v;\lambda ]:=\varphi \lbrack u,v;\lambda ]\in
\Lambda ^{0}(\mathcal{\bar{K}\{}u,v;D_{x}^{-1}\sigma |N\})^{2},\ 
\label{B2.17}
\end{equation}%
owing to Lemma \ \ref{Lm_B2.2}, a priori satisfies the Lax equation \ (\ref%
{B2.10b}). Having assumed that $|\lambda |\ \rightarrow \infty ,$ from the
asymptotic properties of linear equations \ (\ref{B2.16}) and \ (\ref{B2.16a}%
) one obtains the result of Proposition \ \ref{Prop_B2.3}.

\section{\protect\bigskip The two-component polynomial Burgers type
dynamical system integrability analysis}

Proceed now to analyzing the Lax type integrability of the two-component
polynomial Burgers type dynamical system \ (\ref{B1.1}). To do this, owing
to the approach described above, it is necessary to prove to the Lax
equation \ (\ref{B2.10b}) possesses an asymptotic solution \ of the form (%
\ref{B2.13}) in \ $\Lambda ^{0}(\mathcal{\bar{K}\{}u,v;D_{x}^{-1}\sigma
|N\})^{2}.$ Concerning the dynamical system (\ref{B1.1}) the following
proposition holds.

\begin{proposition}
\label{Prop_B3.1}The Lax equation (\ref{B2.10b}) with the differential
matrix operator%
\begin{equation}
K^{\prime ,\ast }[u,v]=\left( 
\begin{array}{cc}
D_{x}^{2}-2uD_{x} & -vD_{x} \\ 
-D_{x} & -uD_{x}%
\end{array}%
\right) ,  \label{B3.1}
\end{equation}%
possesses the asymptotic as $|\lambda |\rightarrow \infty $ solution \ 
\begin{equation}
\varphi (x;\lambda )=(1,1/\lambda )^{\intercal }g(x;\lambda )\exp [-\lambda
^{2}t-\lambda x+D_{x}^{-1}(u+\lambda ^{-1}v)],  \label{B3.1aa}
\end{equation}%
where the scalar invertible local functional element 
\begin{equation}
g(x;\lambda ):=\exp (-u+\sum\limits_{j\in \mathbb{Z}_{+}\backslash
\{0,1\}}D_{x}^{-1}\sigma _{j}[u,v]/\lambda ^{j})\in \Lambda ^{0}(\mathcal{%
\bar{K}\{}u,v\}).  \label{B3.1bb}
\end{equation}%
The solution \ (\ref{B3.1aa}) corresponds to the local conservative quantity 
$\Delta (\lambda ):=\int dx(u+\lambda ^{-1}v)$ $\in \mathcal{D(}\overline{%
\mathcal{A}};dx)$ in the extended ring $\mathcal{\bar{K}\{}%
u,v;D_{x}^{-1}\sigma |N\}^{2}:$%
\begin{equation}
\mathrm{grad}\Delta (\lambda )[u,v]=\varphi (x;\lambda )\in \Lambda ^{0}(%
\mathcal{\bar{K}\{}u,v;D_{x}^{-1}\sigma |N\})^{2}.  \label{B3.1cc}
\end{equation}%
$.$
\end{proposition}

\begin{proof}
\bigskip Assume that the Lax equation (\ref{B2.10b}) possesses the
asymptotic as $|\lambda |\rightarrow \infty $ solution\ (\ref{B2.13}), where 
$\omega (x,t;\lambda )=-\lambda x-\lambda ^{2}t,$%
\begin{equation}
\Lambda ^{0}(\mathcal{\bar{K}\{}u,v;D_{x}^{-1}\sigma |N\})^{2}\ni \varphi
(x;\lambda )=\psi (x,t;\lambda )\exp \{-\ \lambda ^{2}t-\lambda
x+D_{x}^{-1}\sigma (x,t;\lambda )\}  \label{B3.1a}
\end{equation}%
and 
\begin{equation}
\Lambda ^{0}(\mathcal{\bar{K}\{}u,v\})^{2}\ni \psi (x,t;\lambda
)=(1,a(x,t;\lambda ))^{\intercal }  \label{B3.1b}
\end{equation}%
which reduces to an equivalent system of the differential-functional
relationships $\ $%
\begin{equation}
\begin{array}{c}
D_{x}^{-1}\sigma _{t}-\lambda ^{2}+\sigma _{x}+(-\lambda +\sigma
)^{2}-(2u+va)(-\lambda +\sigma )-va_{x}=0, \\ 
a_{t}+a(-\lambda ^{2}+D_{x}^{-1}\sigma _{t})-ua_{x}-(au+1)(-\lambda +\sigma
)=0.%
\end{array}
\label{B3.1c}
\end{equation}%
The coefficients of the corresponding asymptotic expansions $\ $ 
\begin{equation}
\begin{array}{c}
\Lambda ^{0}(\mathcal{\bar{K}\{}u,v\})\ni a(x,t;\lambda )\sim \sum_{j\in 
\mathbb{Z}_{+}}a_{j}[u,v]\lambda ^{-j}, \\ 
\Lambda ^{0}(\mathcal{\bar{K}\{}u,v\})\ni \sigma (x,t;\lambda )\sim
\sum_{j\in \mathbb{Z}_{+}}\sigma _{j}[u,v]\lambda ^{-j},%
\end{array}
\label{B3.2}
\end{equation}%
should satisfy two infinite hierarchies of recurrent relationships%
\begin{equation}
\begin{array}{c}
D_{x}^{-1}\sigma _{j-1,t}+\sigma _{j-1,x}+2\sigma _{j\ }+\sum_{k\in \mathbb{Z%
}_{+}}\sigma _{j-1-k}\sigma _{k}-2u\delta _{j-1,-1}-2u\sigma _{j-1}- \\ 
-a_{j\ }v-v\sum_{k\in \mathbb{Z}_{+}}\sigma _{j-1-k}a_{k}-va_{j-1,x}=0, \\ 
a_{j-2,t}-a_{j\ }+\sum_{k\in \mathbb{Z}_{+}}a_{j-2-k}D_{x}^{-1}\sigma
_{k,t}-ua_{j-2,x}-\delta _{j-2,-1}- \\ 
-\sigma _{j-2}-ua_{j-1}-u\sum_{k\in \mathbb{Z}_{+}}\sigma _{j-2-k}a_{k}=0,%
\end{array}
\label{B3.3}
\end{equation}%
\ compatible for all $j\in \mathbb{Z}_{+}.$ It is easy to calculate \ from (%
\ref{B3.3}) the corresponding coefficients $\ $%
\begin{eqnarray}
\sigma _{0} &=&u,\text{ \ \ \ \ \ }\sigma _{1}=\ u_{x}+v,\text{\ \ \ \ }%
\sigma _{2}=v_{x}+u_{xx}+uu_{x},\text{ \ }  \label{B3.4} \\
\sigma _{3} &=&D_{x}(u^{3}/3+\ uv\ +u_{xx}+2uu_{x}+v_{x}),...,\sigma
_{j}=D_{x}(...),...,  \notag \\
a_{0} &=&0,\text{\ \ \ \ \ \ \ }a_{1}=\ 1,\text{ \ \ }a_{2}=0,\text{ }\
a_{3}=0,...,a_{j}=0,...  \notag
\end{eqnarray}%
and to get convinced that only two functionals 
\begin{equation}
\gamma _{0}:=\int dx\sigma _{0}[u,v]=\int dxu,\gamma _{1}:=\int dx\sigma
_{1}[u,v]=\int dxv,  \label{B3.5}
\end{equation}%
are nontrivial conservation laws with respect to the differentiation $D_{t},$
since all other functionals 
\begin{equation}
\gamma _{j}:=\int dx\sigma _{j}[u,v]=\ \int dxD_{x}(...)=0  \label{B3.6}
\end{equation}%
are trivial$\ $in the ring $\mathcal{\bar{K}\{}u,v\}.$ Equivalently, it
means that the gradient $\varphi (x;\lambda ):=\mathrm{grad}\int
dx(u+\lambda ^{-1}v)=(1,\ 1/\lambda )^{\intercal }\in \Lambda ^{0}(\mathcal{%
\bar{K}\{}u,v\})^{2}$ satisfies the Lax equation \ (\ref{B2.10b}) in the
ring $\mathcal{\bar{K}\{}u,v\}$ and thus \ it should coincide with the
expression \ (\ref{B3.1a}). As a result one easily obtains that 
\begin{equation}
(1,1/\lambda )^{\intercal }=(1,1/\lambda )^{\intercal }g(x;\lambda )\exp
[-\lambda ^{2}t-\lambda x+D_{x}^{-1}(u+\lambda ^{-1}v)],  \label{B3.7}
\end{equation}%
where the scalar invertible element 
\begin{equation}
g(x;\lambda ):=\exp (-u+\sum\limits_{j\in \mathbb{Z}_{+}\backslash
\{0,1\}}D_{x}^{-1}\sigma _{j}[u,v]/\lambda ^{j})\in \Lambda ^{0}(\mathcal{%
\bar{K}\{}u,v\}),  \label{B3.8}
\end{equation}%
giving rise to the expressions \ (\ref{B3.1aa}) and \ (\ref{B3.1bb}). The
latter proves the proposition.
\end{proof}

As a consequence of Proposition \ref{Prop_B3.1} we can formulate the
following theorem.

\begin{theorem}
\label{Tm_B3.2} The two-component polynomial Burgers type dynamical system \
(\ref{B1.1}) possesses only two local conserved quantities $\int dxu$ and $%
\int dxv$ and no other infinite affinely ordered conserved quantities
(either local or nonlocal). Moreover, on the functional manifold $M$ the
Burgers type dynamical system \ (\ref{B1.1}) is linearizable by means of a
Hopf-Cole type transformation and the related linear adjoint mapping to the
matrix Lax type representation%
\begin{equation}
D_{x}\left( 
\begin{array}{c}
f \\ 
\hat{f}%
\end{array}%
\right) =\left[ 
\begin{array}{cc}
(u+\lambda ^{-1}v-\lambda )/2 & 0 \\ 
1 & (\lambda -\lambda ^{-1}v-u\ )/2]%
\end{array}%
\right] \left( 
\begin{array}{c}
f \\ 
\hat{f}%
\end{array}%
\right) ,\text{\ \ \ }  \label{B3.9a}
\end{equation}%
and%
\begin{equation}
\left( 
\begin{array}{c}
f \\ 
\hat{f}%
\end{array}%
\right) =\left[ 
\begin{array}{cc}
\begin{array}{c}
D_{x}(u+\lambda )]/2+ \\ 
+(u+\lambda )(u+\lambda ^{-1}v-\lambda )/2%
\end{array}
& 0 \\ 
(u+\lambda ) & 
\begin{array}{c}
-D_{x}(u+\lambda )]/2+ \\ 
+(u+\lambda )(\lambda -u-\lambda ^{-1}v)/2%
\end{array}%
\end{array}%
\right] \left( 
\begin{array}{c}
f \\ 
\hat{f}%
\end{array}%
\right) ,  \label{B3.9b}
\end{equation}%
compatible for all $\lambda \in \mathbb{C}\backslash \{0\},$ where vector
function $(f,\hat{f})^{\intercal }\in \Lambda ^{0}(\mathcal{\bar{K}\{}%
u,v;D_{x}^{-1}\sigma |N\})^{2},\ $ $N=$ $2.$
\end{theorem}

\begin{proof}
Based on Proposition \ref{Prop_B3.1} and recent results of \cite{POS}, one
can apply to the two component polynomial Burgers type dynamical system \ (%
\ref{B1.1}) the following generalized Hopf-Cole type linearizing
transformation:%
\begin{equation}
\mathcal{\ }u:=D_{x}\ln \tilde{g}(x;\lambda ),  \label{B3.10}
\end{equation}%
where have put 
\begin{equation}
\tilde{g}(x;\lambda ):=g(x;\lambda )^{-1\ }\exp (\lambda ^{2}t\ +\lambda x\
-\ \lambda ^{-1}D_{x}^{-1}v)\in \Lambda ^{0}(\mathcal{\bar{K}\{}%
u,v;D_{x}^{-1}\sigma |N\})  \label{B3.11}
\end{equation}%
and, by construction, $\ $it should be set $N=2.$ Having substituted \ (\ref%
{B3.10}) into \ (\ref{B1.1}), one finds easily the following system of
linear equations%
\begin{equation}
D_{t}\tilde{g}=D_{x}^{2}\tilde{g}+v\tilde{g},\text{ \ \ \ \ \ \ \ \ \ }D_{x}%
\tilde{g}=\ u\ \tilde{g},  \label{B3.12a}
\end{equation}%
which easily reduces to the following system of differential relationships: 
\begin{equation}
D_{t}\tilde{f}=\ (u_{x}+u^{2}\ +v)\ \tilde{f}/2=0,\ \ D_{x}\tilde{f}=\
(u/2)\ \tilde{f},  \label{B3.12b}
\end{equation}%
\ if to make the change of variables $\tilde{f}:=$ $\tilde{g}^{1/2}\in
\Lambda ^{0}(\mathcal{\bar{K}\{}u,v;D_{x}^{-1}\sigma |N\}).$ The system \ (%
\ref{B3.12b}) can be specified further, if to make use of the substitution $%
\tilde{f}:=$ $f\exp (\lambda x-\lambda ^{-1}D_{x}^{-1}v),$ \ giving rise to
the next scalar operator Lax type representation 
\begin{equation}
D_{t}f=\ \ D_{x}(u+\lambda )\ f/2\ +\ (u\ +\lambda )D_{x}f,\text{ \ \ }%
D_{x}f=\ (u+\lambda ^{-1}v-\lambda )\ f/2,  \label{B3.12c}
\end{equation}%
compatible for all $\lambda \in \mathbb{C}\backslash \{0\}.$

Now we will proceed to constructing a suitably linearly extended adjoint
differential relationships \cite{PrY} \ for the system of equations \ (\ref%
{B3.12c}). Doing the standard way, one easily obtains that the following
linearly adjoint relationship compatible with the second equation of \ (\ref%
{B3.12c}) 
\begin{subequations}
\begin{equation}
\begin{array}{c}
D_{x}\hat{f}=D_{x}(\frac{\hat{f}\text{ }f}{f}\ )=-f^{-1}\hat{f}D_{x}f\text{
\ }+\text{ }f^{-1}D_{x}(\hat{f}\text{ }f)= \\ 
\\ 
=-[(\lambda -\lambda ^{-1}v-u\ )/2]\hat{f}+f^{-1}D_{x}(\hat{f}\text{ }%
f)=-[(\lambda -\lambda ^{-1}v-u\ )/2]\hat{f}+\chi \lbrack u,v;\lambda ]f,%
\end{array}
\label{B3.13a}
\end{equation}%
holds, \ where we have put, by definition, that $\ D_{x}(\hat{f}$ $f):=\chi
\lbrack u,v;\lambda ]f^{2}$ for some arbitrarily chosen element $\ \chi
\lbrack u,v;\lambda ]\in $ $\mathcal{\bar{K}\{}u,v\}.$ The compatibility of
\ (\ref{B3.13a}) with the first equation of \ (\ref{B3.12c}) and its
suitable extension gives rise to $\ $the \ condition$\ \chi \lbrack
u,v;\lambda ]=1.\ $\ Thus, \ we have obtained that the linearly adjoint
relationship compatible with the second equation of \ (\ref{B3.12c}) reads
as 
\end{subequations}
\begin{equation}
D_{x}\hat{f}=-\ (\lambda -\lambda ^{-1}v-u\ )\ \hat{f}/2+f.  \label{B3.13b}
\end{equation}%
The respectively adjoint linear relationship compatible with the first
equation of \ (\ref{B3.12c}) obtains easily as 
\begin{equation}
D_{t}\hat{f}=-\ D_{x}(u+\lambda )\hat{f}/2+\ (u\ +\lambda )D_{x}\hat{f},%
\text{ \ \ }  \label{B3.13c}
\end{equation}%
and is compatible with \ (\ref{B3.13b}) for all $\lambda \in \mathbb{C}%
\backslash \{0\}.$ It is now dexterous to rewrite \ equations (\ref{B3.12c}%
), \ (\ref{B3.13b}) and \ (\ref{B3.13c}) as the following two equivalent
matrix systems: 
\begin{equation}
D_{x}\left( 
\begin{array}{c}
f \\ 
\hat{f}%
\end{array}%
\right) =\left( 
\begin{array}{cc}
(u+\lambda ^{-1}v-\lambda )/2 & 0 \\ 
1 & (\lambda -\lambda ^{-1}v-u\ )/2]%
\end{array}%
\right) \left( 
\begin{array}{c}
f \\ 
\hat{f}%
\end{array}%
\right) ,\text{ \ \ \ }  \label{B3.14a}
\end{equation}%
and

\begin{equation}
D_{t}\left( 
\begin{array}{c}
f \\ 
\hat{f}%
\end{array}%
\right) =\left[ \left( 
\begin{array}{cc}
\begin{array}{c}
D_{x}(u+\lambda )]/2+ \\ 
+(u+\lambda )(u+\lambda ^{-1}v-\lambda )/2%
\end{array}
& 0 \\ 
(u+\lambda ) & 
\begin{array}{c}
-D_{x}(u+\lambda )]/2+ \\ 
+(u+\lambda )(\lambda -u-\lambda ^{-1}v)/2%
\end{array}%
\end{array}%
\right) \right] \left( 
\begin{array}{c}
f \\ 
\hat{f}%
\end{array}%
\right) ,\text{ \ \ \ }  \label{B3.14b}
\end{equation}%
where, by construction, elements $(f,\hat{f})^{\intercal }\in \Lambda ^{0}(%
\mathcal{\bar{K}\{}u,v;D_{x}^{-1}\sigma |N\})^{2},\ $ $N=$ $2.$ Based on the
systems (\ref{B3.14a}) and (\ref{B3.14b}) one can easily calculate that 
\begin{equation}
D_{x}(\hat{f}\ f)=f^{2},\text{ \ \ \ }D_{t}(\hat{f}\text{ }f)=(u+\lambda
)f^{2},  \label{B3.15}
\end{equation}%
Since the mutual compatibility condition of relationships \ (\ref{B3.15})
reduces to the expression 
\begin{equation}
D_{t}f=[D_{x}(u+\lambda )/2]f+(u+\lambda )D_{x}f,  \label{B3.16}
\end{equation}%
exactly coinciding with the first equation of the system \ (\ref{B3.14b}), \
we now can interpret both systems (\ref{B3.14a}) and (\ref{B3.14b}) as the
corresponding matrix Lax type representation for the Burgers type system \ (%
\ref{B1.1}). This proves the theorem.
\end{proof}

\bigskip\ 

The scalar representation \ (\ref{B3.12a}), as it can be easily observed,
can be generalized to the following higher-order evolution equation:%
\begin{equation}
D_{t}\tilde{g}=D_{x}^{n}\tilde{g}+v\tilde{g}=0,\text{ \ \ }D_{x}\tilde{g}=\
u\ \tilde{g}\ ,  \label{B3.17}
\end{equation}%
where $n\in \mathbb{N}\backslash \{1,2\},\tilde{g}\in \Lambda ^{0}(\mathcal{%
\bar{K}\{}u,v;D_{x}^{-1}\sigma |N\})$ and there is imposed no \textit{a
priori} constraint on the function $v\in \mathcal{\bar{K}\{}u,v\}$ except
the functional $\int dxv\in \mathcal{D(}\overline{\mathcal{A}};dx)$ has to
be a conserved quantity with respect to the differentiation $D_{t}.$
Applying to \ (\ref{B3.17}) the nonlocal change of variables $u:=2D_{x}\ln 
\tilde{f}[u,v;\lambda ]$ for $\tilde{f}\in \Lambda ^{0}(\mathcal{\bar{K}\{}%
u,v;D_{x}^{-1}\sigma |N\}),$ $N=2,$ one can obtain a new infinite hierarchy
of two-component integrable polynomial Burgers type dynamical systems,
generalizing the systems studied before in \ \cite{Ta,POS}. \ For instance,
at $n=3$ \ we find the following dynamical Burgers-Korteweg-de Vries type
dynamical system of the third order:%
\begin{eqnarray}
\ u_{t} &=&\ u_{3x}+3\ (uu_{x})_{x}+\ 3u^{2}u_{x}+v_{x},  \label{B3.18} \\
v_{t} &=&(u\ r[u,v])_{x},  \notag
\end{eqnarray}%
where $r[u,v]\in \mathcal{\bar{K}\{}u,v\}$ is for the present an arbitrary
element. \ To choose from them those for which the dynamical systems of type
\ (\ref{B3.18}) will possess suitably extended matrix Lax type
representations, it is natural to take the first pair of Lax type equation \
(\ref{B3.9a}) 
\begin{eqnarray}
D_{x}f &=&\ (u+\lambda ^{-1}v-\lambda )\ f/2,\text{ \ \ \ }  \label{B3.19a}
\\
&&  \notag \\
\text{\ }D_{x}\hat{f} &=&-\ (u+\lambda ^{-1}v-\lambda )\ \hat{f}/2+f  \notag
\end{eqnarray}%
for $(f,\hat{f})^{\intercal }\in \Lambda ^{0}(\mathcal{\bar{K}\{}%
u,v;D_{x}^{-1}\sigma |N\})^{2},\ $ $N=$ $2,$ \ and to supplement it by means
of the following systems of evolution equations, naturally generalizing \
that of (\ref{B3.14b}) with respect to the temporal parameters $t_{n}\in 
\mathbb{R}$ $:$%
\begin{eqnarray}
D_{t_{n}}f &=&\ D_{x}\alpha _{n}(x;\lambda )\text{ }f/2+\ \alpha
_{n}(x;\lambda )D_{x}f,  \label{B3.19b} \\
&&  \notag \\
D_{t_{n}}\hat{f} &=&-\ D_{x}\alpha _{n}(x;\lambda )\text{ }\hat{f}/2\ +\
\alpha _{n}(x;\lambda )D_{x}\hat{f},  \notag
\end{eqnarray}%
which \ are, by construction, compatible for all $\lambda \in \mathbb{C}%
\backslash \{0\}$ for a polynomial in $\lambda \in \mathbb{C}$ element $%
\alpha _{n}(x;\lambda )\in \mathcal{\bar{K}\{}u,v\},$ $n\in \mathbb{N},$
satisfying the standard determining \ relationship 
\begin{equation}
D_{t_{n}}(u+\lambda ^{-1}v-\lambda )\ =D_{x}[D_{x}^{\ \ }\alpha
_{n}(x;\lambda )\ +\ \ (u+\lambda ^{-1}v-\lambda )\alpha _{n}(x;\lambda )].
\label{B3.20}
\end{equation}%
It is also easy to check that for $n=1$ the choice 
\begin{equation}
\alpha _{1}(x;\lambda )=u+\lambda \text{ \ \ }  \label{B3.21}
\end{equation}%
entails exactly the Burgers type dynamical system \ (\ref{B1.1}).

The general algebraic structure of the whole infinite hierarchy of resulting
dynamical systems \ (\ref{B3.20}) can be extracted easily from the matrix
spectral Lax pair \ (\ref{B3.14a}) 
\begin{equation}
D_{x}\left( 
\begin{array}{c}
f \\ 
\hat{f}%
\end{array}%
\right) =\left( 
\begin{array}{cc}
(u+\lambda ^{-1}v-\lambda )/2 & 0 \\ 
1 & (\lambda -u-\lambda ^{-1}v)/2%
\end{array}%
\right) \left( 
\begin{array}{c}
f \\ 
\hat{f}%
\end{array}%
\right) :=l[u,v;\lambda ]\left( 
\begin{array}{c}
f \\ 
\hat{f}%
\end{array}%
\right) ,  \label{B3.22}
\end{equation}%
which allows by means of the gradient-holonomic scheme \cite{PM,BPS} to \
obtain successfully from the corresponding differential commutator equation 
\begin{equation}
D_{x}S=[l,S],\text{ \ \ \ }S=\left( 
\begin{array}{cc}
S_{11} & S_{12} \\ 
S_{21} & \ S_{22}%
\end{array}%
\right) ,  \label{B3.23}
\end{equation}%
for the related "monodromy" matrix $S:=S(x;\lambda )\in sl(2;\mathbb{C)}$ \
the resulting \cite{No,PM} canonical differential relationships \ for the
gradient $(\varphi _{1},\varphi _{2})^{\intercal }:=\varphi :=$\textrm{grad}$%
(\mathrm{tr}S)\in \Lambda ^{0}(\mathcal{\bar{K}\{}u,v;D_{x}^{-1}\sigma
|N\})^{2}$ of the dynamical $D_{x}$ and $D_{t}$-invariant trace functional $%
\mathrm{tr}$ $S(x;\lambda )\in \mathcal{D(}\overline{\mathcal{A}};dx):$%
\begin{equation}
\left\{ 
\begin{array}{c}
-D_{x}\varphi _{1}\ +\ D_{x}^{-1}uD_{x}\varphi _{1}+\
D_{x}^{-1}vD_{x}\varphi _{2}=\lambda \varphi _{1}, \\ 
\varphi _{1}=\lambda \varphi _{2},%
\end{array}%
\right.   \label{B3.34}
\end{equation}%
where the component $\varphi _{1}\in \mathcal{\bar{K}\{}u,v\}$ possesses,
owing to the construction, the following differential-algebraic
representation: 
\begin{equation}
\varphi _{1}=\ (u+\lambda ^{-1}v-\lambda )S_{21}+D_{x}S_{21}  \label{B3.35}
\end{equation}%
for some polynomial expression $S_{21}:=S_{21}(x;\lambda )\in \mathcal{\bar{K%
}\{}u,v\}.$ The differential expressions \ (\ref{B3.34}) can be rewritten in
the following useful matrix form:%
\begin{equation}
\Lambda \varphi =\lambda \varphi ,\text{ \ \ }\Lambda :=\left( 
\begin{array}{cc}
-D_{x}+D_{x}^{-1}uD_{x} & D_{x}^{-1}vD_{x} \\ 
1 & 0%
\end{array}%
\right) ,\text{\ }  \label{B3.36}
\end{equation}%
where the recursion operator $\Lambda :T^{\ast }(\mathcal{\bar{K}\{}%
u,v\})\rightarrow T(\mathcal{\bar{K}\{}u,v\})$ satisfies the determining
operator equation 
\begin{equation}
D_{t}\Lambda =[\Lambda ,K^{\prime ,\ast }],  \label{B3.37}
\end{equation}%
easily following from the Noether-Lax condition \ (\ref{B2.10b}) and the
adjoint linear spectral relationship \ (\ref{B3.36}).

Recall now that our Burgers type dynamical system \ (\ref{B1.1}) possesses
only two conservations laws: $\gamma _{0}=\int dxu$ and $\gamma _{1}=\int
dxv\in $ $\mathcal{D(}\overline{\mathcal{A}};dx).$ This means that the
expression \ (\ref{B3.35}) exactly equals $\varphi _{1}=$\textrm{grad}$%
_{u}\gamma _{0}[u,v]=1,$ or equivalently the condition 
\begin{equation}
\ D_{x}[D_{x}S_{21}(x;\lambda )+\ (u+\lambda ^{-1}v-\lambda
)S_{21}(x;\lambda )]=0  \label{B3.37a}
\end{equation}%
should be satisfied for some element $S_{21}(x;\lambda )\in \Lambda ^{0}(%
\mathcal{\bar{K}\{}u,v\})\ $\ and all $\lambda \in \mathbb{C}\backslash
\{0\}.$ The following proposition characterizes asymptotic as $|\lambda
|\rightarrow \infty \ \ $solutions to \ (\ref{B3.37a}) and their
relationships to the generalized dynamical systems \ (\ref{B3.20}).

\begin{proposition}
The nonnegative degree polynomial part of the asymptotic, as $|\lambda
|\rightarrow \infty ,$ solution $S_{21}(x;\lambda )\sim \sum_{j\in \mathbb{Z}%
_{+}}\lambda ^{-j}S_{21}^{(j)}[u,v;\lambda ]$ \ to the differential
relationship \ (\ref{B3.37a}) makes it possible to represent the generating
elements $\alpha _{n}(x;\lambda )\in \Lambda ^{0}(\mathcal{\bar{K}\{}u,v\})$
\ of the generalized dynamical systems \ (\ref{B3.20}) as 
\begin{equation}
\alpha _{n}(x;\lambda )=\ (\lambda ^{n}S_{21}(x;\lambda ))_{+}\ \ 
\label{B3.37b}
\end{equation}%
for any $n\in \mathbb{Z}_{+}.$
\end{proposition}

\begin{proof}
Taking into account that the whole hierarchy of the generalized Burgers type
dynamical systems \ (\ref{B3.20}) can be represented \ in the recursive form 
\begin{equation}
D_{t_{n}}(u,v)^{\intercal }=\Phi ^{n}(D_{x}u,D_{x}v)^{\intercal },\text{ \ \ 
}\Phi :=\Lambda ^{\ast }=\left( 
\begin{array}{cc}
\ D_{x}+D_{x}^{\ }uD_{x}^{-1} & \text{ \ \ }1 \\ 
D_{x}^{\ }vD_{x}^{-1} & \text{ \ \ \ }0%
\end{array}%
\right) ,  \label{B3.38}
\end{equation}%
we can rewrite it equivalently as 
\begin{equation}
D_{t_{n}}(u+\lambda ^{-1}v-\lambda )=D_{x}[D_{x}^{\ \ }\alpha _{n}(x;\lambda
)\ \ +\ (u+\lambda ^{-1}v-\lambda )\alpha _{n}(x;\lambda )],  \label{B3.39}
\end{equation}%
where, by definition, $\ $%
\begin{equation}
\alpha _{n}(x;\lambda ):=(\lambda ^{n}\alpha (x;\lambda ))_{+}
\label{B3.39a}
\end{equation}%
is the corresponding nonnegative degree polynomial part generated by the
asymptotic solution $\ \alpha (x;\lambda )\sim \sum_{j\in \mathbb{Z}%
_{+}}\lambda ^{-j}\alpha _{j}[u,v;\lambda ]$ as $|\lambda |\rightarrow
\infty $ to the differential functional equation 
\begin{equation}
D_{x}^{2}\alpha (x;\lambda )+D_{x}((u+\lambda ^{-1}v-\lambda )\alpha
(x;\lambda ))=0,  \label{B3.40}
\end{equation}%
exactly equivalent to the dual to \ (\ref{B3.36}) symmetry relationship \ 
\begin{equation}
\Phi \text{ }(D_{x}\alpha ,D_{x}\beta )^{\intercal }=\lambda (D_{x}\alpha
,D_{x}\beta )^{\intercal }  \label{B3.41}
\end{equation}%
with the generalized symmetry of the flow \ (\ref{B1.1}) 
\begin{equation}
(D_{x}\alpha ,D_{x}\beta )^{\intercal }:=\sum_{j\in \mathbb{Z}_{+}}\lambda
^{-j}\Phi ^{j}(D_{x}u,D_{x}v)^{\intercal }.  \label{B3.42}
\end{equation}%
The observation that the differential functional equation \ \ (\ref{B3.40})
coincides exactly with that of \ (\ref{B3.37a}) proves the proposition.
\end{proof}

From Theorem \ref{Tm_B1.1} we also can derive that the Burgers type
dynamical system \ (\ref{B1.1}) does not allow on the functional manifold $M$
a Hamiltonian formulation. This means that the recursion operator \ (\ref%
{B3.36}) constructed above and found before in \cite{Fo,Ma} for the
two-component Burgers type dynamical system \ (\ref{B1.1}), proves to be not
factorizable by means of suitably defined compatible Poissonian structures,
as they on the whole, eventually do not exist. It is strongly related with
the fact that the dynamical system \ (\ref{B1.1}) possesses no infinite
hierarchy of local conservation laws, whose existence is responsible for the
factorization mentioned above. Nonetheless, similar to the situation
happened in the work \ \cite{POS}, if one to succeed to state that the
Burgers type dynamical system \ (\ref{B1.1}) does possess another infinite
hierarchy of \textit{nonlocal} conservation laws, then some degree of the
found before symmetry recursion operator \ (\ref{B3.38}) will be already
factorized by means of the respectively constructed Poissonian structures.
Yet, by now, this problem remains still open.

\section{Conclusion}

Having based on the differential-algebraic approach \cite%
{BPS,PM,PAPP,PAPP1,POS} to testing the Lax type integrability of nonlinear
dynamical systems on functional manifolds, we stated that the two component
polynomial Burgers type dynamical system \ (\ref{B1.1}) does possess an
adjoint matrix Lax type representation and the corresponding recursion
operator, which does not allow a bi-Poissonian  factorization and makes it
possible to construct only two local conserved quantities. A problem to
construct a generalized bi-Poissonian  factorization of a suitably powered
recursion operator, similarly to that of the work \ \cite{POS}, is left for
the future analysis. Thus, the differential-algebraic approach, jointly with
considerations based on the symplectic geometry, can serve as simple enough
as effective tool for analyzing the Lax type integrability of a wide class
of polynomial nonlinear dynamical systems on functional manifolds. Moreover,
as it was recently demonstrated in \cite{PAPP2}, this approach also appears
to be useful in the case of nonlocal polynomial dynamical systems.

\bigskip

\section{Acknowledgements}

D.B. acknowledges the National Science Foundation (Grant CMMI-1029809), A.P.
cordially thanks Prof. J. Cie\'{s}li\'{n}skiemu (Bia\l ystok University,
Poland),  Prof. I. Mykytyuk (Pedagogical University of Krakow, Poland) and
Prof. Prof. A. Augustynowicz (Gdansk University, Poland) for useful
discussions of the results obtained. A.P., E.O. and K.S.   gratefully
acknowledge partial support of the research in this paper from the
Turkey-Ukrainian: TUBITAK-NASU Grant 110T558. \

\end{document}